\documentclass{article}
\usepackage{graphics}

\usepackage{amssymb}
\usepackage{amsmath}
\usepackage{amsthm}
\usepackage{color}

\graphicspath{{figures/}} 

\newtheorem{theorem}{Theorem}
\newtheorem{lemma}[theorem]{Lemma}
\newtheorem{corollary}[theorem]{Corollary}
\newtheorem{observation}[theorem]{Observation}
\newtheorem{proposition}[theorem]{Proposition}

\newcommand{\slope}{\ensuremath{\textrm{sl}}}
\newcommand{\adj}{\textrm{adj}}
\newcommand{\seq}[1]{\langle #1 \rangle}
\newcommand{\CS}{\Sigma}
\newcommand{\Paths}{{\cal P}}
\newcommand{\Graphs}{{\cal G}}
\newcommand{\Dirs}{{\cal V}}
\newcommand{\Ords}{{\cal O}}
\newcommand{\vl}{\phi}
\newcommand{\dir}{\vec{v}}

\begin{document}

\title{Monotone Simultaneous Embeddings of Directed Paths\thanks{%
Research supported by the ESF EUROCORES
programme EuroGIGA - ComPoSe, Austrian Science Fund (FWF): I
648-N18 and grant EUI-EURC-2011-4306.
T.H.\ supported by the Austrian Science Fund (FWF): P23629-N18 `Combinatorial Problems on Geometric Graphs'.
A.P.\ is recipient of a DOC-fellowship of the Austrian
Academy of Sciences at the Institute for Software Technology, Graz University
of Technology, Austria. 
}
}

%

\author{Oswin Aichholzer\thanks{Institute for Software Technology, Graz University of Technology, Austria,
\texttt{[oaich|thackl|apilz|bvogt]@ist.tugraz.at}.}
\and Thomas~Hackl$^\dagger$
\and Sarah~Lutteropp\thanks{Institute of Theoretical Informatics, Karlsruhe Institute of Technology, Germany,
\texttt{sarah.lutteropp@student.kit.edu, mched@iti.uka.de}. A part of this research has been accomplished when the authors were visiting  Graz University of Technology.}
\and Tamara~Mchedlidze$^\ddagger$
\and Alexander~Pilz$^\dagger$
\and Birgit~Vogtenhuber$^\dagger$}

\maketitle

\begin{abstract}
We study monotone simultaneous embeddings of upward planar digraphs, which are simultaneous embeddings where the drawing of each digraph is upward planar, and the directions of the upwardness of different graphs can differ.
We first consider the special case where each digraph is a directed path. 
In contrast to the known result that any two directed paths admit a monotone simultaneous embedding, there exist examples of three paths that do not admit such an embedding for any possible choice of directions of monotonicity. 
We prove that if a monotone simultaneous embedding of three paths exists then it also exists for any possible choice of directions of monotonicity.
We provide a polynomial-time algorithm that, given three paths, decides whether a monotone simultaneous embedding exists and, in the case of existence, also constructs such an embedding. 
On the other hand, we show that already for three paths, any monotone simultaneous embedding might need a grid of exponential (w.r.t.\ the number of vertices) size.
For more than three paths, we present a polynomial-time algorithm that, given any number of paths and predefined directions of monotonicity, decides whether the paths admit a monotone simultaneous embedding with respect to the given directions, including the construction of a solution if it exists.
Further, we show several implications of our results on monotone simultaneous embeddings of general upward planar digraphs.
Finally, we discuss complexity issues related to our problems.
\end{abstract}

\pagenumbering{arabic}

\section{Introduction}
Let $\{G_i=(V,E_i)| 1\leq i \leq k\}$ be a set of $k$ distinct planar graphs sharing the same vertex set.
A \emph{simultaneous embedding} of these graphs is a set of their planar drawings $\{\Gamma_i | 1\leq i \leq k\}$ such that each vertex of $V$ is represented by the same point in the plane in each of the drawings.
Simultaneous embeddings were introduced as a model for visual comparison of different relations of the same object set, as well as for a visualization of dynamic changes of a single relation. 
Depending on how the edges in a drawing must be realized, simultaneous embeddings are distinguished as follows.
In \emph{geometric simultaneous embeddings} all edges are required to be straight segments.
In \emph{simultaneous embeddings with fixed edges}, there is no special restriction on the shape of the edges, but the common edges of the graphs are required to be drawn in the same way in each of the drawings.
Finally, a simultaneous embedding in which there is no restriction on how the edges are drawn is simply called a \emph{simultaneous embedding}.
Bl\"asius, Kobourov, and Rutter~\cite{gdhandbook13} give
an extensive overview of the known results for these types of simultaneous embeddings.

Simultaneous embeddings were also studied for upward planar digraphs.
A directed graph (\emph{digraph}, for short) is called \emph{upward planar} if it admits a planar drawing so that all its edges are represented by curves, monotonically increasing in a common direction, which is traditionally called \emph{upward}.
Upward drawings are motivated by a desire of a clearer expression of a hierarchy among a set of objects.
An \emph{upward simultaneous embedding} of $k$ upward planar digraphs is a set of upward planar drawings of given graphs, such that each vertex is represented by the same point in the plane in each of the drawings.
However, the choice of direction of ``upwardness'' to be common for all graphs does not make any sense.
It is easy to see that for any two graphs $G_1=(V,E_1)$ and $G_2=(V,E_2)$, an upward simultaneous embedding with only one direction of upwardness does not exist if $G=(V, E_1 \cup E_2)$ contains a directed cycle.
Motivated by this simple fact, Giordano et al.~\cite{book_embeddings} considered upward simultaneous embeddings where the directions of upwardness are different. They showed that any two upward planar digraphs admit an upward simultaneous embedding, where the directions of upwardness differ by $\pi/2$. Giordano, Liotta, and Whitesides~\cite{GiordanoLW08} gave a characterization of upward simultaneous embeddable digraphs with respect to the same direction.
Pampel~\cite[p.~71]{pampel} considered upward simultaneous embeddings of directed paths where the directions of upwardness are different and only a subset of vertices is common to all paths. This problem is known as \textsc{Strictly Monotone Trajectory Drawing}. 
Pampel~\cite[p.~71]{pampel}  showed that the problem is NP-hard even for paths with four vertices.

In this paper we study upward simultaneous embeddings for more than two graphs and different directions of upwardness. We formalize the problem as follows.
Let $\dir$ be a vector in $\mathbb{R}^2$. A drawing of a directed graph $G=(V,E)$  is called \emph{$\dir$-monotone}, if its edges are represented by (directed) curves which are monotonically increasing in the direction of $\dir$.
Let $\Dirs=\{\dir_1,\dots,\dir_k\}$ ($k > 1$) be a set of vectors and let $\Graphs=\{G_1, G_2, \dots, G_k\}$ be a set of $k$ distinct upward planar digraphs sharing the same vertex set.
A $\Dirs$-\emph{monotone simultaneous embedding} of $\Graphs$ is a set $\{\Gamma_1,\dots,\Gamma_k\}$ of planar drawings such that:
(1) $\Gamma_i$ is a $\dir_i$-monotone drawing of $G_i$ ($1 \leq i \leq k$) and
(2) an equally labeled vertex is represented by the same point in all drawings of the sequence.
A \emph{monotone simultaneous embedding} of $\Graphs$ is a $\Dirs$-monotone simultaneous embedding for some set $\Dirs$ of vectors. Observe that if a monotone simultaneous embedding of $\Graphs$ exists, it is naturally associated with a vector set $\Dirs$. However, these vectors are not required to be radially ordered around the origin. In case of monotone simultaneous embeddings of a sequence of graphs, this will be a requirement. But we return to this point in the next section, after explaining how the sequences and sets of paths are related.    We remark that determining whether a monotone simultaneous embedding exists is more restricted than the \textsc{Strictly Monotone Trajectory Drawing} problem, since in the former each vertex is shared by all of the paths.    

We study the following problem. Given a set $\Graphs$ of upward planar digraphs, we ask whether it admits a monotone simultaneous embedding. In order to shed light on this problem we also consider the following constrained version: Given a set $\Graphs$ of upward planar digraphs and a set $\Dirs$ of vectors, with $|\Graphs| = |\Dirs|$, we study whether $\Graphs$ admits a $\Dirs$-monotone simultaneous embedding. 

Recall that in monotone simultaneous embeddings, as defined above, there is a restriction on how the edges are drawn, that is, they are required to be represented by monotone curves in some direction. This fact makes monotone simultaneous embeddings completely different from simultaneous embeddings of undirected graphs, for which it is known that any number of planar graphs admits a simultaneous embedding~\cite{PW98}. As we have already mentioned, for monotone simultaneous embeddings, this is not the case, and existence of a monotone simultaneous embedding strongly depends on the choice of the directions of monotonicity. Intuitively, it is clear that the choice of such directions becomes more restricted as the order among the vertices of graphs becomes more strict. We first look at the core of this problem by assuming
that each of our graphs is a simple directed spanning path (i.e., directed from one end of the path to the other) of the common vertex set $V$. Then we prove several implications of our results on general upward planar digraphs. 

We start with introducing some tools in Section~\ref{sec:prel}, namely relations to circular sequences of point sets and the dual representation of the considered problem.
Our main results are concentrated in Sections~\ref{sec_easy}-\ref{sec_hard} and are as follows. 

\begin{itemize}
\item In Section~\ref{sec_three} 
we consider sets of three directed paths. In contrast to sets of two paths which always admit a monotone simultaneous embedding (
	see also~\cite{Brass2007117,book_embeddings}), there exist sets of three paths which do not admit such an embedding~\cite{Asinowski2008}.
We show that, if a monotone simultaneous embedding for three directed paths exists, then it also exists for any set of predefined directions. We also show that this result is tight with respect to the number of paths, i.e., it does not hold for four or more paths.
Further, we provide an example of three paths for which any monotone simultaneous embedding requires a grid of exponential (w.r.t.\ the number of vertices) size.
\item In Section~\ref{sec_lp} 
we consider larger sets of paths. We show that, given any set $\Paths=\{P_1,\dots,P_k\}$ of paths and a set $\Dirs=\{\dir_1,\dots,\dir_k\}$ of vectors, we can decide in polynomial time whether $\Paths$ admits a $\Dirs$-monotone simultaneous embedding, including the construction of a solution if it exists. 
Together with the results from Section~\ref{sec_three}, this implies that, for $k=3$, the existence of a monotone simultaneous embedding (without predefined directions) can be decided in polynomial time, which answers a question posed by Asinowski in~\cite{Asinowski2008}. In case of existence, such an embedding can be constructed in polynomial time as well. 
\item In Section~\ref{sec:implications} 
we derive several implications of the aforementioned results on upward planar digraphs. Based on work of Giordano et al.~\cite{GiordanoLW08}, we first show that, given a set of $k$ upward planar digraphs, the question whether they admit a monotone simultaneous embedding can be reduced to the same question for $k$ paths, under the condition that the order of the vertices of the graphs in the projection on their direction of monotonicity is fixed.  We then state several results on monotone simultaneous embeddings of upward planar digraphs that are implied by this fact.  
\item Finally, in Section~\ref{sec_hard} 
we discuss the complexity of monotone simultaneous embeddings of $k$ paths (on the same vertex set $V$) when the directions of monotonicity are not predefined.
We show that the construction problem becomes intractable for $k>3$ and give reasons why the complexity of the decision problem does not directly follow from the NP-hardness of deciding stretchability of pseudo-line arrangements~\cite{shor}.
Also, we consider a generalization of monotone simultaneous embeddings of paths and show NP-hardness of this version of the problem.
\end{itemize}
We conclude in Section~\ref{sec_conc} with several open problems. 

\section{Preliminaries}
\label{sec:prel}

\subsection{Relation to circular sequences}
\label{sec:circ_seq}

The problem of monotone simultaneous embedding is strongly related to the \emph{circular sequence} of a point set (see Goodman and Pollack~\cite{non_circular,semispaces} for details).
The circular sequence of a point set was also used in the related work by Giordano, Liotta, and Whitesides~\cite{GiordanoLW08}.  
Let $\ell$ be any line and $S = \{s_1, \dots, s_n\}$  be a labeled set of points. 
Consider the orthogonal projection of $S$ on $\ell$.
This gives a permutation of the indices of the points.
Continuously rotate~$\ell$ counterclockwise.
Every time a supporting line of two points becomes normal to $\ell$, two indices change their position (we omit details concerned with collinear point triples and parallel supporting lines).
After having rotated $\ell$ by $180^\circ$, the initial permutation of indices is reversed.
Every pair of indices changed their relative position exactly once.
This sequence of permutations defines the \emph{circular sequence} of a point set $S$, which we denote by $\CS(S)$.%
\footnote{The circular sequence is infinite. However, a half-period, which corresponds to a rotation of $\ell$ by $180^\circ$, completely determines the sequence.}
An arbitrary periodic sequence of index permutations (not connected to any point set) which fulfills these properties (i.e., every pair changes its relative position exactly once per half-period) 
is called an \emph{allowable sequence}.
Hence, a circular sequence is an allowable sequence that stems from a projection of a point set on a rotating line.
In our problem we are given a set of permutations (paths) and the question is whether there exists a point set such that its circular sequence contains the given set of permutations%
\footnote{Since we consider only a half-period $\CS$ of a circular sequence, we allow that, instead of the original, the reversed permutation appears in $\CS$.}.
In a related work, Asinowski~\cite{Asinowski2008} considered similar questions for allowable sequences: 
He introduced \emph{suballowable sequences}, which are subsequences of allowable sequences, and investigated their properties; see also the remarks in Section~\ref{sec_easy}. 
The following observations stem from properties of both, circular and allowable, sequences and most of them can also be found in (or derived from)~\cite{Asinowski2008}.

\begin{observation}
\label{obs:circular}
A set $\{P_1,\dots,P_k\}$ of directed paths  on a common set of $n$ vertices admits a monotone simultaneous embedding if and only if there exists a set $S$ of $n$ labeled points, such that for each $i=1,\dots,k$, the circular sequence $\CS(S)$ contains a permutation of indices defined either by $P_i$ or by its reverse path.  
\end{observation}

Note that only for a circular sequence $\CS$ one can construct a point set $S$ such that $\CS=\CS(S)$. 
This fact gives us a necessary condition for the existence of monotone simultaneous embeddings. 
This necessary condition is the target of the remainder of this section, for which we need further preparation.
Consider a set $\Paths$ of paths on a vertex set~$V=\{v_1,\dots,v_k\}$. 
We denote by $I(i,P)$ the rank of $v_i$ in path $P \in \Paths$.
A set of vectors which are all directed to the same half-plane is called \emph{adjusted}. 
We denote by $\Paths_{ij}^{\adj}$ the set of paths which results from reversal of some of the paths of $\Paths$ so that $I(i,P) < I(j,P)$ for each $P\in \Paths$. If $\Paths=\Paths^{\adj}_{ij}$ for some indices $i,j$, $1\leq i,j \leq n$, then $\Paths$ is called \emph{adjusted}. 
Recalling Observation~\ref{obs:circular}, and noticing that in an adjusted set of paths some indices $i$ and $j$ do not switch their position, we get the following.

\begin{observation}
\label{obs:adjusted}
If an adjusted set of $k$ paths admits a $\{\dir_1,\dots,\dir_k\}$-monotone simultaneous embedding, then vectors
$\{\dir_1,\dots,\dir_k\}$ are also adjusted.  
\end{observation}

Until now we only talked about sets of paths. In order to relate a set of paths to a circular sequence we need to consider an order among the elements of a set of paths. Thus, we will denote by $\seq{\Paths}$ an ordering of the elements of $\Paths$. Analogously to circular sequences, if for any triple of paths $P_a$, $P_b$, $P_c \in \seq{\Paths}$, where $a < b < c$, there exists a pair of vertices $v_i, v_j \in V$ with $I(i,P_a) < I(j,P_a)$ and $I(i, P_b) > I(j,P_b)$ for which it holds that $I(i,P_c) < I(j,P_c)$, we say that the sequence $\seq{\Paths}$ of paths is \emph{non-allowable}. It is \emph{allowable} otherwise. More generally, a set $\Paths$ of paths is called \emph{allowable} if its elements can be ordered to form an allowable sequence of paths.

\begin{proposition}\label{alowable_set}
Let $\Paths$ be an adjusted set of $k$ paths on the same set of $n$ vertices.  An allowable sequence $\seq{\Paths}$ of paths -- if it exists -- is unique up to reversal and can be constructed within $O(kn^2)$ time. 
\end{proposition}
\begin{proof}
The fact that, if there exists $\seq{\Paths}$ which is an allowable sequence of paths then it is unique up to reversal, follows from the basic properties of allowable sequences (see~\cite{Asinowski2008, non_circular}). 

In the following, we provide a constructive algorithm 
similar to the one discussed in~\cite[page 4751]{Asinowski2008} in order to be able to analyze its running time. 

Let $i$ and $j$ be two arbitrary indices. We partition~$\Paths$ into two sets $\Paths_1 = \{P \in \Paths : I(i,P) > I(j,P)\}$ and $\Paths_2 = \{P \in \Paths : I(i,P) < I(j,P)\}$. In an allowable sequence $\seq{\Paths}$ the elements of $\Paths_1$ (resp. $\Paths_2$) appear consecutive. Thus, an allowable $\seq{\Paths}$ is either $\seq{\Paths_1}$ concatenated with $\seq{\Paths_2}$ or the other way around. We assume that the former happens, in the latter case the arguments are similar.    

Consider two indices~$k$ and~$l$, that change their relative position among the elements of~$\Paths_1$ (resp. $\Paths_2$). Then they cannot change it in $\Paths_2$ (resp. $\Paths_1$) as well.
The collection $\Paths_1$ (resp. $\Paths_2$) can be partitioned recursively by the pair $k,~l$.
Since the indices $k$ and $l$ change their relative position among the elements of $\Paths_1$ (resp. $\Paths_2$), we have to arrange the resulting two partitions in a way that the second (resp. first) partition have the same relative position of $k$ and $l$ as the elements of $\Paths_2$ (resp. $\Paths_1$).
If at some recursive step we can not apply this partition operation, this implies that no allowable sequence $\seq{\Paths}$ exists. 

For the running time of this algorithm, note that in $O(kn)$ time, all paths can be preprocessed such that checking whether $I(i,P) < I(j,P)$ can be done in constant time. Then each step of the above recursion takes $O(k)$ time. Together with the fact that there are $O(n^2)$ pairs of indices that can switch their position, we obtain a total running time of $O(kn + kn^2) = O(kn^2)$.
\end{proof}

We are now ready to state a necessary condition for the existence of monotone simultaneous embeddings that stems from the properties of circular sequences. 
\begin{lemma}
\label{lemma:allowable}
Let $\Paths$ be a set of $k$ paths on a set of $n$ vertices. If $\Paths$ admits a monotone simultaneous embedding then for every pair of indices $i,j$, $1\leq i,j\leq n$, the set $\Paths_{ij}^{\adj}$ is allowable.
\end{lemma}
\begin{proof}
Assume that $\Paths$ admits a monotone simultaneous embedding. We will show that for any possible choice of indices  $1\leq i,j \leq n$, the set $\Paths_{ij}^{\adj}$ is allowable. Let $S$ be a set of points representing vertices of the paths in $\Paths$, in a $\{\dir_1,\dots,\dir_k\}$-monotone simultaneous embedding of $\Paths$. Let $i$ and $j$ be two arbitrary indices. Rotate $S$ and vectors $\{\dir_1,\dots,\dir_k\}$ until the line supporting $i$ and $j$ is vertical and $i$ appears before $j$ in the positive y-direction. Let $S'$ be the resulting point set. Reverse those vectors of $\{\dir_1,\dots,\dir_k\}$ that point in the negative $y$-direction and also the corresponding paths of $\Paths$. Let $\{\dir_1',\dots,\dir_k'\}$ and $\Paths'$ denote the resulting vectors and paths, respectively. In $\Paths'$, index $i$ appears before index $j$ in each of the paths. Thus we have  that $\Paths'=\Paths_{ij}^{\adj}$. Moreover, for any path $P_a \in \Paths'$, the permutation defined by $P_a$ is an element of $\CS(S')$ (the corresponding element of $\CS(S')$ is created when points of $S'$ are projected on the line given by $\dir_a'$).  Thus, the constructed $\Paths_{ij}^{\adj}$ is allowable. 
\end{proof}

Lemma~\ref{lemma:allowable} and Proposition~\ref{alowable_set} imply the following. 
\begin{corollary}
\label{cor:set_to_sequence}
If a set $\Paths$ of $k$ paths with $n$ vertices admits a monotone simultaneous embedding, then for every $i,j$, $1\leq i,j\leq n$, there exists an ordering $\seq{\Paths_{ij}^{\adj}}$ of $\Paths_{ij}^{\adj}$ which is allowable. The sequence $\seq{\Paths_{ij}^{\adj}}$ can be computed within $O(kn^2)$ time.
\end{corollary}

We observe that Corollary~\ref{cor:set_to_sequence} allows us to restrict considerations to allowable adjusted sequence of paths. Any question which can be resolved in polynomial time for an allowable adjusted sequence of $k$ paths, is also polynomial-time solvable for a set of $k$ paths.  
 
Recall that, if an allowable sequence $\seq{P_1,\dots, P_k}$ of paths admits a 
monotone simultaneous embedding, where $P_i$ is embedded monotone in direction $\dir_i$ for $1 \leq i \leq k$, then the sequence $\seq{\dir_1,\dots,\dir_k}$ of vectors is ordered around the origin. Thus, for $\seq{P_1,\dots, P_k}$ it makes sense to talk about \emph{$\seq{\dir_1,\dots,\dir_k}$-monotone simultaneous embeddings}, where the vectors appear in this order around the origin. Moreover, a \emph{monotone simultaneous embedding of a sequence of paths} is refined to be a $\seq{\dir_1,\dots,\dir_k}$-monotone simultaneous embedding. 
To emphasize the difference with the previous definition, in monotone simultaneous embeddings of a set of paths, no order was required on the vectors.

\subsection{The Dual Problem}\label{sec_dual}
We will work on the problem of monotone simultaneous embedding in the dual plane by using the standard duality transform where a point $s = (x_s, y_s)$ is mapped to a non-vertical line $l: y = x_sx-y_s$ and vice versa.
See~\cite{edelsbrunner} for properties of the dual transform.
In this section we briefly recall some properties of the transform which are used in this paper;
see \figurename~\ref{fig_duality} for an example.
We denote by $S$ a set of points in the primal (or a set of lines in the dual setting, respectively).
Similarly, in the primal setting, single points are denoted by~$s$, and lines are denoted by~$l$.
As a special case, we denote vertical lines in the dual by~$\vl$. 
Notation-wise, we do not distinguish between the primal and the dual setting. For example, $s$~is a point
in the primal and at the same time a line in the dual.

\begin{figure}[htb]
\centering
\includegraphics{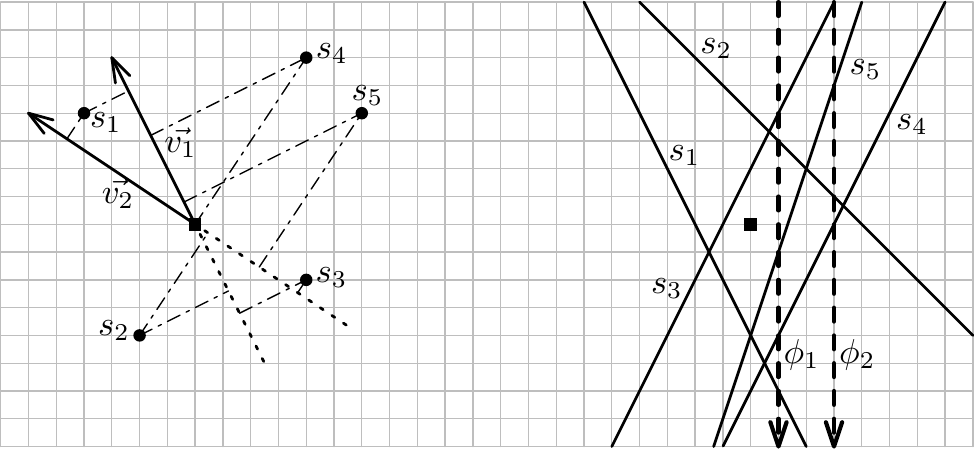}
\caption{A point set (left) and its dual arrangement of lines (right).
The two vectors to the left are represented by the two vertical lines in the dual.
Observe the order of projection of the points on the vectors in the primal and the corresponding order of the intersections with the vertical lines in the dual.}
\label{fig_duality}
\end{figure}

A well-known property of the duality transform is that the primal point~$s$ is below a primal line~$l$ if and only if the dual point $l$ is below the dual line $s$.
Further, the infinite set of points on the vertical line $\vl: x = \alpha$ corresponds to the set of lines with slope~$\alpha$.
Let $S$ be a primal set of $n$ points and let $\pi$ be the order in which the points are traversed by translating a line of slope $\alpha$ below all points in $S$ to a line above all points of $S$.
Then $\pi$ is also the order in which the dual lines of $S$ intersect the vertical line $\vl: x = \alpha$ in negative $y$-direction.
Consider now the lines that are normal to any line $\ell$ of slope $\alpha$, i.e., the ones of slope $(-1/\alpha)$.
The traversal of $S$ by a line of that slope corresponds to the order in which the points of $S$ are projected on $\ell$.
Therefore, this is the order in which the dual lines intersect the vertical line $\vl': x = (-1/\alpha)$. Observe also that, for $\alpha \rightarrow 0$, the order is given by the slope of the lines. So, we have the following.


\begin{observation}\label{obs:vec_vert}
	Consider a point set $S$, a directed path~$P$ containing the points of $S$, and a vector $\dir$ of slope $\alpha$. The edges of $P$ can be drawn resulting in a $\dir$-monotone drawing of $P$ if and only if the dual lines of $S$ intersect the vertical line $\vl: x = (-1/\alpha)$ in the same order as they appear either along~$P$ or along its reverse.
\end{observation}

The vector $\dir$ and the vertical line $\vl$ in the previous observation are said to correspond to each other. Recall that in the monotone simultaneous embedding problem we are given a set $\{P_1,\dots,P_k\}$ of paths on the same vertex set $V$ and our goal is to determine the positions of vertices $V$, such that the edges of each $P_i$ can be drawn resulting in its $\dir_i$-monotone drawing for some vector $\dir_i$, $i=1,\dots,k$. Recall also that Corollary~\ref{cor:set_to_sequence} allows us to consider only allowable adjusted sequence of paths. The following observations will be the repeatedly used in the remainder of the paper.

\begin{observation}
\label{obs:dual_problem}
An adjusted allowable sequence $\seq{P_1,\dots,P_k}$ of paths on a common set of vertices admits a monotone simultaneous embedding if and only if there exists a sequence of $k$ ordered vertical lines $\vl_1,\dots,\vl_k$ and a set  $\{s_1,\dots,s_n\}$ of $n$ non-vertical lines that for each $i=1,\dots,k$, the lines $\{s_1,\dots,s_n\}$ intersect the vertical line $\vl_i$ in the same order as the points $\{s_1,\dots,s_n\}$ appear along $P_i$.
\end{observation}

For the constrained version of monotone simultaneous embeddings we have the following. 

\begin{observation}
\label{obs:dual_problem_constr}
An adjusted allowable sequence $\seq{P_1,\dots,P_k}$ of paths on a common set of vertices admits a $\seq{\dir_1,\dots,\dir_k}$-monotone simultaneous embedding if and only if there exists a set of $n$ non-vertical lines $\{s_1,\dots,s_n\}$ that for each $i=1,\dots,k$, intersect the vertical line $\vl_i$ in the same order as $\{s_1,\dots,s_n\}$ appear along $P_i$, where $\vl_i$ is the vertical line in the dual corresponding to the vector $\dir_i$, $i=1,\dots,k$.
\end{observation}


\section{Monotone Simultaneous Embeddings in the Dual}\label{sec_easy}

\subsection{Two and Three Paths}\label{sec_three}
It is well known that given any sequence $\seq{P_1, P_2}$ of two paths and two vectors $\seq{\dir_1,\dir_2}$, there always exists a $\seq{\dir_1,\dir_2}$-monotone simultaneous embedding of $\seq{P_1, P_2}$; see for example~\cite{gdhandbook13, Brass2007117, book_embeddings}.
To give some intuition for the more complex cases of three or more paths, we present the following approach which utilizes the dual setting.

Let $\vl_1$ and $\vl_2$ be the two vertical lines along which the order for the paths is defined in the dual.
Let the dual lines be labeled in increasing order of appearance in $P_1$, and let $I(i,P_2)$ be the index of $s_i$ in $P_2$ (the function $I$ gives us the order $\pi$ used before in the form of indices).
Let the line $s_i$ pass through the point $(x_{\vl_1}, -i)$ and $(x_{\vl_2}, -I(i, P_2))$. 
This gives us a primal point set allowing a $\seq{\dir_1,\dir_2}$-monotone simultaneous embedding of $\seq{P_1, P_2}$.

In contrast to two paths, it is not always possible to find a monotone simultaneous embedding of three paths, even if there are no constraints on the directions of monotonicity. This also shows that the necessary condition of Lemma~\ref{lemma:allowable} is not sufficient. 
The proof of the following theorem uses a classical result by Ringel~\cite{ringel} on the stretchability of pseudo-line arrangements based on Pappus' Theorem.

\begin{theorem} {\bf (\cite[Proposition 8]{Asinowski2008})}
There are allowable path sets of three or more paths that do not admit a monotone simultaneous embedding. 
\end{theorem}

On the positive side, a monotone simultaneous embedding for three paths does not strongly depend on the choice of the vectors: 

\begin{theorem}\label{thm_3_always}
	Suppose that a path set $\Paths=\{P_1,P_2,P_3\}$ admits a $\{\dir_1,\dir_2,\dir_3\}$-monotone simultaneous embedding. Then  $\Paths$ admits a $\{\dir'_1,\dir'_2,\dir'_3\}$-monotone simultaneous embedding, for any vectors $\dir'_1,\dir'_2,\dir'_3$, provided that their radial order around the origin is the same as the one of $\dir_1,\dir_2,\dir_3$. 
\end{theorem}

\begin{proof}  
By Corollary~\ref{cor:set_to_sequence}, we can assume that sequence $\seq{P_1,P_2,P_3}$ is allowable and adjusted and hence, also the sequence $\{\dir_1,\dir_2,\dir_3\}$ is adjusted.
Further, we can assume without loss of generality that $\dir'_1$ is a vector with slope $0$.
By Observation~\ref{obs:dual_problem}, the dual vertical lines $\vl_1,\vl_2,\vl_3$, appear in this order from left to right. 
Further, the dual lines $s_1,\dots,s_n$ that correspond to the vertices, cross each $\vl_i$ in the order as they appear on~$P_i$. 
Recall the sphere model of the projective plane.
Consider a plane and a sphere in $E^3$ such that the center point is not in the plane.
Every line that intersects the plane in one point and passes through the center of the sphere intersects the sphere in two antipodal points, and every line through the center of the sphere intersects the plane, except if the line is parallel to the plane.
The union of the points on the sphere defined in this way by lines parallel to the plane represents the line at infinity~$\ell_\infty$.
Hence, we are given a bijective mapping from every point in the plane to two antipodal points on the sphere (not on~$\ell_\infty$).
In this mapping, a line in the plane corresponds to a great circle on the sphere.
Consider now the vertical lines  $\vl_1,\vl_2,\vl_3$ that are crossed by the non-vertical lines $s_1,\dots,s_n$.
If we apply the projective transformation that corresponds to rotating the sphere such that the great circle corresponding to $\vl_1$ equals~$\ell_\infty$, we obtain different lines $\vl'_2,\vl'_3,s_1',\dots,s_n'$ in the plane, corresponding to $\vl_2,\vl_3,s_1,\dots,s_n$, respectively. The order of intersections of the transformed lines $s_1',\dots,s_n'$ on $\vl'_i$, $i \in \{2,3\}$, is identical to the order of intersections of the original lines $s_1,\dots,s_n$ on $\vl_i$. Furthermore, the order of the slopes of the transformed lines $s_1',\dots,s_n'$ is identical to the order of intersections of $s_1,\dots,s_n$ with~$\vl_1$. Equivalently, $\vl'_1$ is the vertical line $\vl'_1: x = -\infty$ and hence, $\vl'_1$ corresponds to a vector $\dir'_1$ with slope $0$ in the primal.
Now we can scale and translate the transformed lines $\vl'_2,\vl'_3,s_1',\dots,s_n'$ such that we have the vectors $\dir'_2$ and $\dir'_3$ in any position we want.
\end{proof}

In contrast to the statement of Theorem~\ref{thm_3_always}, with four paths, we can first encounter the situation where the actual slopes and not just the relative radial order of the vectors influence the existence of a monotone simultaneous embedding. In other words, the statement of Theorem~\ref{thm_3_always} is tight with respect to the number of paths; see Proposition~\ref{prop:four_paths} 
in Section~\ref{sec_lp}. 

Moreover, even if a monotone simultaneous embedding of three paths exists, an exponential (in the number of vertices) ratio of the smallest and largest distance between vertices of the embedding might be unavoidable.

\begin{proposition}\label{prop_three_expo}
There exists a set of three paths such that every monotone simultaneous embedding needs a grid whose size is exponential in the number of vertices of the paths.
\end{proposition}
\begin{proof}
Consider the example shown in \figurename~\ref{fig_three-expo}.
Let the number of dual lines be any natural number $n=3m+2$, $m>0$.
The three specified paths are 
\begin{eqnarray*}
	&& P_1  =  \seq{s_1, s_2, s_3, s_4, s_5, s_6, s_7, s_8, s_9, s_{10}, s_{11}, \ldots}, \\
	&& P_2  =  \seq{s_1} \circ \seq{s_3, s_6, s_9 \ldots} \circ \seq{s_2} \circ \seq{s_4, s_5, s_7, s_8, s_{10}, s_{11}, \ldots}, \quad \mbox{ and} \\
	&& P_3  =  \seq{\ldots, s_{11}, s_9, s_8, s_6, s_5, s_3} \circ \seq{s_2} \circ \seq{\ldots, s_{10}, s_7, s_4} \circ \seq{s_1}.\footnotemark%
\end{eqnarray*}
\footnotetext{Here we use $\circ$ as the concatenation operator, e.g., $\seq{s_i}\circ\seq{s_j}=\seq{s_i,s_j}$.}
The lines $s_1,\dots, s_n$ cross the vertical lines $\vl_1,\vl_2,\vl_3$ in the order indicated by the paths $P_1, P_2,P_3$, respectively.  In the figure only $\vl_2$ and $\vl_3$ are shown.
W.l.o.g., we assume unit distance between $\vl_2$ and $\vl_3$ and that $s_1$ has slope $0$.
\begin{figure}[htb]
\centering
\includegraphics{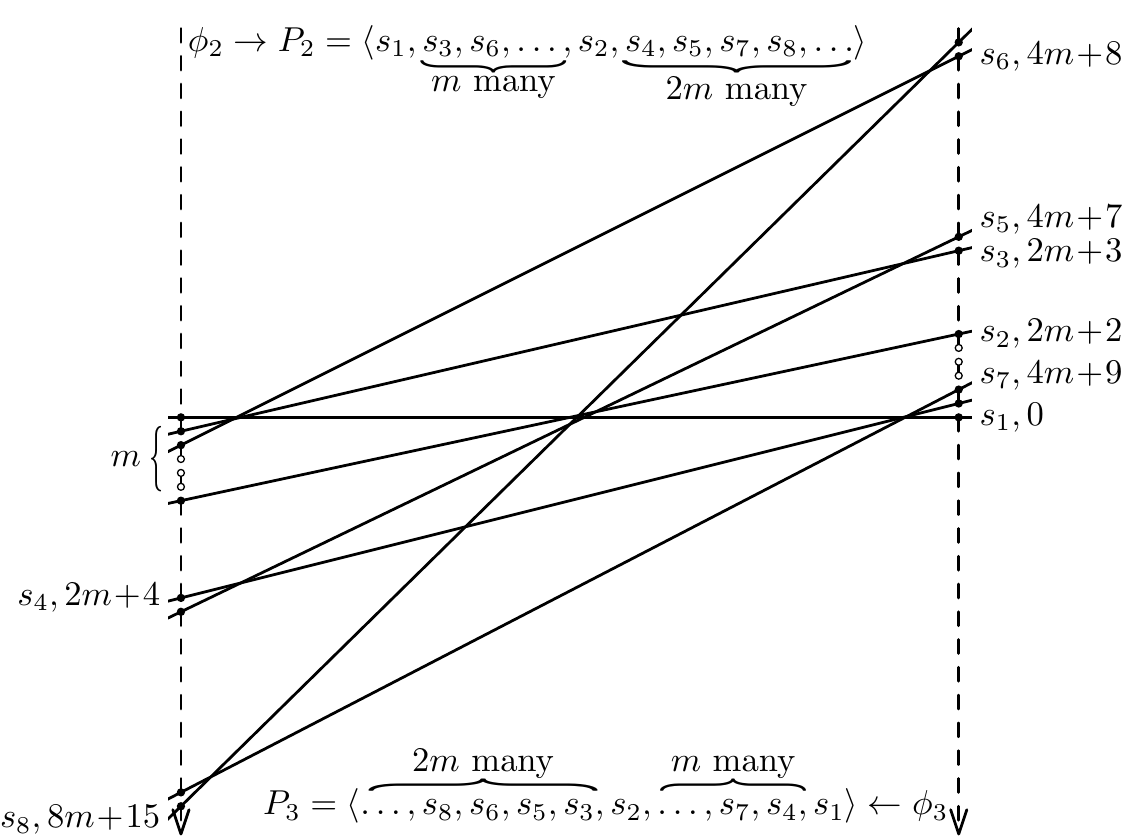}
\caption{Example for three paths where any monotone simultaneous embedding requires a grid of exponential size. 
	The dual lines are labeled with their index (e.g., $s_2$) and their slope (e.g., $2m+2$).}
\label{fig_three-expo}
\end{figure}

I order to satisfy the order of $s_1,\dots,s_n$ on $\vl_1$ (at the ``far left'' of the example), the order of the slopes $\slope(s_i)$ of the lines $s_i$ has to be $\slope(s_i)<\slope(s_j)$, for $i<j$.
Because of the given order on $\vl_2$ there are $m$ lines intersecting $\vl_2$ between the intersections of $s_1$ and $s_2$ with $\vl_2$.
The same is true on $\vl_3$.
Assuming at least unit distance between two intersections on a vertical line, $\slope(s_2)$ has to be at least $2m+2$.
As $\slope(s_2)$ has to be strictly smaller than $\slope(s_3)$, $\slope(s_3)$ is at least $2m+3$ and analogously $\slope(s_4)$ is at least $2m+4$.
The line $s_5$ has to intersect $\vl_2$ below the intersection of $s_4$, and the intersection of $s_5$ with $\vl_3$ has to be above the intersection with $s_3$.
As a result, $\slope(s_5)\geq \slope(s_3)+\slope(s_4)$, i.e., the slope of $s_5$ is at least $4m+7$.
In general, $\slope(s_{3i})\geq \slope(s_{3i-1})+1$, $\slope(s_{3i+1})\geq \slope(s_{3i-1})+2$, and $\slope(s_{3i+2})\geq 2\!\cdot\!\slope(s_{3i-1})+3-2(i-1)$, for $1\leq i\leq m$.
Therefore, we get for the biggest slope on $n=3m+2$ lines, $\slope(s_n)\geq 2^m(2m+3)+2m-1=\Omega(n2^{\frac{n}{3}})$.
This proves that the ratio between minimal distance (unit distance) and maximal distance in the constructed point set is exponential in $n$. 
\end{proof}

\subsection{Fixed Vectors}\label{sec_lp}

In this section we restrict considerations to the setting where we are not only given a set of paths, but also have predefined directions of monotonicity.  

\begin{theorem}\label{thm:poly}
Given a set ${\dir_1, \ldots, \dir_k}$ of vectors and a set $\Paths=\{P_1,\dots, P_k\}$ of paths on the same set of $n$ vertices,
it can be decided in polynomial time if a $\{\dir_1, \ldots, \dir_k\}$-monotone simultaneous embedding of $\Paths$ exists. If such an embedding exists, it can also be constructed in polynomial time.
\end{theorem}

\begin{proof}
By Corollary~\ref{cor:set_to_sequence}, if $\Paths$ admits a monotone simultaneous embedding, an allowed adjusted ordering $\seq{\Paths}$ of $\Paths$ exists and can be computed in polynomial time. 
W.l.o.g.\ we can assume that $\seq{\dir_1,\dots, \dir_k}$ is the ordering of adjusted vectors corresponding to $\seq{\Paths}$.
 
Let $\vl_1,\dots,\vl_k$ be the vertical lines in the dual plane corresponding to the vectors in $\seq{\dir_1,\dots, \dir_k}$.  By Observation~\ref{obs:dual_problem_constr}, in order to check whether $\seq{\Paths}$ admits a $\seq{\dir_1,\dots, \dir_k}$-monotone simultaneous embedding we need to check whether there exist non-vertical lines $s_1,\dots,s_n$ that cross $\vl_i$ as indicated by $P_i$, $i=1,\dots,k$.  Let $y_{i,j}$ be the $y$-coordinate of the intersection of the line $\vl_i$ with the line $s_j$.
Then, for every path $P_i$ and every pair $(s_l, s_m)$ of neighbored vertices in $P_i$, we have the constraint $y_{i,l} > y_{i,m}$, or, equivalently, $y_{i,l} \geq y_{i,m} + 1$ (since any solution can be scaled along the $x$-axis).
Further, let $q_i$ be the distance between the vertical lines $\vl_i$ and $\vl_{i+1}$.
To ensure that $s_1,\ldots, s_n$ are straight lines, we have the constraint $(y_{2,j} - y_{1,j})/q_1 = (y_{(i+1),j} - y_{i,j})/q_i)$ for all $1\leq j \leq n$ and all $2\leq i < k$.
To conclude, we observe that the constructed linear program can be solved in polynomial time.
\end{proof}

Note that this result does not contradict Proposition~\ref{prop_three_expo}: Although any monotone simultaneous embedding of the example from Figure~\ref{fig_three-expo} needs a grid of exponential (w.r.t.\ the number of vertices) size, the resulting coordinates still admit a representation of only polynomial size. 

In~\cite{Asinowski2008}, Asinowski asked whether deciding realizability of a suballowable sequence with three terms, which, in our terminology, is equivalent to the existence of a monotone simultaneous embedding for three paths, is a tractable problem. 
Combining Theorem~\ref{thm_3_always} with Theorem~\ref{thm:poly}, we can answer this question in the affirmative.

\begin{corollary}
\label{corol:threepaths}
Given a set of three paths, it can be decided in polynomial time whether they admit a monotone simultaneous embedding. If such an embedding exists, it can be constructed in polynomial time.
\end{corollary}

Consider again the linear program in the proof of  Theorem~\ref{thm:poly}. If the directions of the monotonicity are not provided as a part of the input, i.e., the distances $q_i$ are variables instead of constants, the presented encoding results in a quadratically constrained program. Thus, the program from the proof of Theorem~\ref{thm:poly} does not provide a means for answering the decision question for the existence of a monotone simultaneous embedding of paths in polynomial time.
Moreover, the following proposition suggests that deciding the existence of a monotone simultaneous embedding might be harder for $k\geq4$ than it is for $k=3$; see also Section~\ref{sec_hard}. 

\begin{proposition}
\label{prop:four_paths}
There exists a set $\Paths=\{P_1,P_2,P_3,P_4\}$ of four paths and two sets of vectors $\{\dir_1,\dir_2,\dir_3,\dir_4\}$ and $\{\dir'_1,\dir'_2,\dir'_3,\dir'_4\}$ with the same radial order around the origin, such that $\Paths$ admits a $\{\dir_1,\dir_2,\dir_3,\dir_4\}$-monotone simultaneous embedding, but does not admit a $\{\dir'_1,\dir'_2,\dir'_3,\dir'_4\}$-monotone simultaneous embedding. 
\end{proposition}
\begin{proof}
\begin{figure}[htb]
\centering
\includegraphics{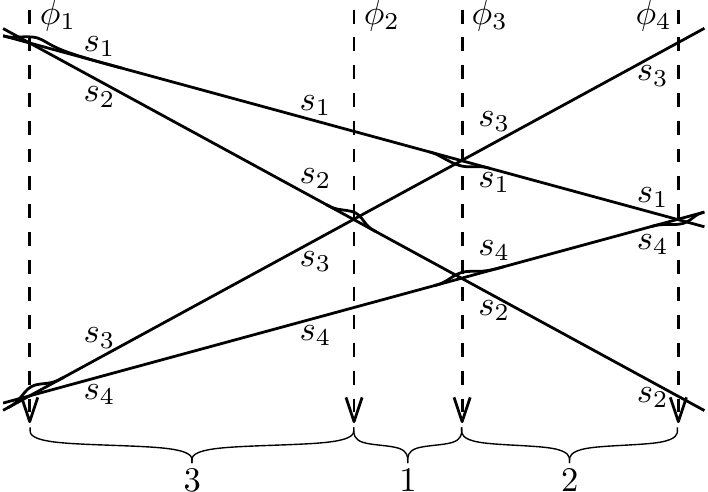}
\caption{The embeddability of this example depends on the relative distance between the vertical lines.
The drawing shows the limit case, the bends in the line show the intended order of intersection.
}
\label{fig_infeasible_snapshots}
\end{figure}
Consider the example shown in \figurename~\ref{fig_infeasible_snapshots}.
If we move the vertical lines $\vl_2$ and $\vl_3$ towards each other while leaving $\vl_1$ and $\vl_4$ unchanged,
then the linear program does not have a solution and hence, no simultaneous embedding is possible.
\end{proof}

\section{Implications for Upward Planar Digraphs}
\label{sec:implications}
Up to this point we have considered monotone simultaneous embeddings of directed paths. In this section we show how our results can be applied to 
provide insight on monotone simultaneous embeddings of more general graph families. More specifically, we consider upward planar digraphs, and upward planar digraphs with Hamiltonian paths. We list several definitions and known results before stating the main result of this section.
An \emph{$st$-digraph} is a biconnected acyclic digraph with exactly one source $s$ and one sink $t$. A \emph{planar $st$-digraph} is an $st$-digraph that is planar and embedded in the plane with vertices $s$ and $t$ on the boundary of the external face.\footnote{In the definition of a (planar) $st$-digraph in~\cite{GiordanoLW08}, the edge $(s,t)$ is required to be an edge of the $st$-digraph. As in our case $s$ and $t$ share a common face, $(s,t)$ can always be added. Therefore, the relevant results from~\cite{GiordanoLW08} (restated in Lemma~\ref{lemma:includingst} and Theorem~\ref{theorem:upse}) apply also to our setting.}
A \emph{topological numbering} of a digraph is an assignment of numbers to its 
vertices, such that for every directed edge $(v,v')$, the number assigned
to $v'$ is greater than the number assigned to~$v$. If each vertex is assigned a distinct number, then we talk about \emph{topological ordering} (see~\figurename~\ref{fig_st_graph}). Every acyclic digraph has at least one topological ordering, and such an ordering can be computed in $O(n)$ time. 

\begin{figure}[htb]
\centering
\includegraphics{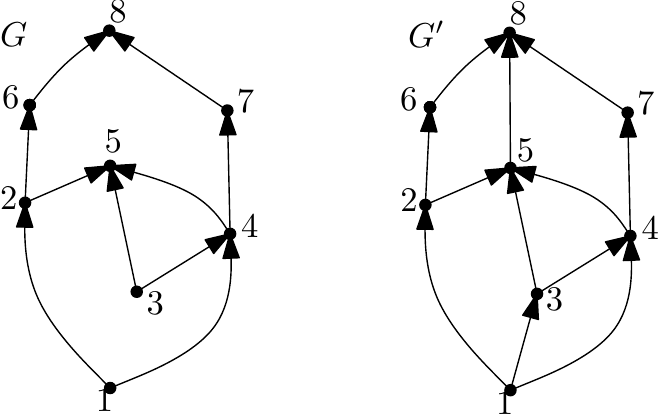}
\caption{A digraph $G$ with a topological ordering $\rho$ defined by the labels at the vertices and an including planar $st$-digraph $G'$ of $G$ that preserves $\rho$.}
\label{fig_st_graph}
\end{figure}

Let $G$ be an upward planar digraph. A planar $st$-digraph that includes $G$ as a spanning subgraph is called \emph{including planar $st$-digraph} of $G$. While every upward planar digraph $G$ can be augmented to a planar $st$-digraph $G'$ that has $G$ as a subgraph (see e.g.~\cite{DiBattistaT88}), it is not generally true that an upward planar digraph $G$ with a topological ordering $\rho$ can be augmented to a planar $st$-digraph $G'$ such that $G'$ still fulfills the topological ordering $\rho$ of $G$ (see~\figurename~\ref{fig_st_graph}).
If the latter is true, we say that $G'$ \emph{preserves} $\rho$. However, as the following lemma by Giordano et al.~\cite{GiordanoLW08} states, it can be tested efficiently whether $G'$ exists. 

\begin{lemma}{\bf (\cite[Lemma 5]{GiordanoLW08})}
\label{lemma:includingst}
Let $G$ be an upward planar digraph with $n$ vertices and $\rho$ be its topological ordering. There exists an $O(n)$-time algorithm that tests whether there exists an including planar $st$-digraph of $G$ that preserves $\rho$.
\end{lemma}

Let again $G$ be an upward planar digraph and $\dir$ be an arbitrary vector. In a $\dir$-monotone drawing of $G$, its vertices, when orthogonally projected on $\dir$,
imply a topological ordering of $G$. However, a general upward planar graph can have exponentially many distinct topological orderings. This explains the additional freedom contained in the problem of monotone simultaneous embeddings for upward planar digraphs, since each digraph can appear in such an embedding in many different ways, according to the number of its topological orderings.

Thus, to generalize our results for paths to upward planar digraphs,
we have to consider the restricted case where each upward planar digraph is provided together with a topological ordering. 
More precisely, let $G$ be an upward planar digraph and $\rho$ its topological ordering.  
A $\dir$-monotone drawing of $G$, such that the order in which its vertices appear in the orthogonal projection on $\dir$ coincides with $\rho$, is called ($\dir,\rho)$-\emph{monotone}.
More generally, an upward drawing of $G$ is called \emph{$\rho$-constrained} if it is $(\dir,\rho)$-monotone for some vector $\dir$.
In~\cite{GiordanoLW08,MchedlidzeS09}, $\rho$-constrained drawings of upward planar digraphs were considered for so-called \emph{book embeddings}, drawings where vertices are required to lie on a single oriented line and edges are represented by curves monotonically increasing in the direction of this line. In~\cite{GiordanoLW08} $\rho$-constrained book embeddings were utilized to construct a so-called \emph{upward point set embedding} of an upward planar digraph with a given \emph{mapping}, that is an upward planar drawing where the positions of the vertices of the graph are provided as a part of the input. The following theorem is a restricted version of \cite[Theorem~3]{GiordanoLW08}. 

\begin{theorem}[\cite{GiordanoLW08}]
\label{theorem:upse} Let $G$ be an upward planar digraph with $n$ vertices, $\rho$ be a topological ordering of $G$, and $\dir$ be a vector. Assume that the vertices of $G$ are positioned in the plane such that their orthogonal projection on $\dir$ coincides with $\rho$.
Then, 
the edges of $G$ can be drawn as polygonal lines resulting in a $(\dir,\rho)$-monotone planar drawing
if and only if $G$ has an including planar $st$-digraph $G'$ preserving~$\rho$. Also, such a $(\dir,\rho)$-monotone drawing of $G$ (with at most $2n-3$ bends per edge) can be computed in $O(n^2)$. 
\end{theorem}

Let $\Graphs=\{G_1,\dots,G_k\}$ be a set of upward planar digraphs with topological orderings $\Ords=\{\rho_1,\dots,\rho_k\}$, respectively, and let ${\vec{\Dirs}}=\{\dir_1,\dots,\dir_k\}$ be a set of vectors. We say that a simultaneous embedding $\Gamma$ of $\Graphs$ is \emph{$(\Dirs,\Ords)$-monotone} if the drawing of $G_i$ implied by $\Gamma$ is $(\dir_i,\rho_i)$-monotone, for each $i=1,\dots,k$. More generally, a simultaneous embedding of $\Graphs$ is $\Ords$-constrained, if it is $(\Dirs,\Ords)$-monotone for some vectors $\Dirs$.  If $G=(V,E)$ is an upward planar digraph and $\rho$ is a topological ordering of $G$, the directed path $P$ with vertex set $V$, which traverses the vertices in the order indicated by the topological ordering $\rho$, is said to be \emph{implied by} $\rho$. Now we are ready to state the main result of this section.

\begin{theorem}
\label{theorem:equiv}
Let $\Graphs=\{G_1,\dots,G_k\}$ be a set of upward planar digraphs on the same set of $n$ vertices and let $\Dirs=\{\dir_1,\dots,\dir_k\}$ be a set of vectors. Consider $\Ords=\{\rho_1,\dots,\rho_k\}$ and $\Paths=\{P_1,\dots,P_k\}$, where $\rho_i$ is a topological ordering of $G_i$ and $P_i$ is the directed path implied by $\rho_i$, for $i=1,\dots,k$. 

The set $\Graphs$ admits a $(\Dirs,\Ords)$-monotone simultaneous embedding if and only if the set $\Paths$ of paths  admits a $\Dirs$-monotone simultaneous embedding, and each $G_i$, $i=1,\dots,k$, has an including planar $st$-digraph $G'_i$ that preserves $\rho_i$. In case of existence, a $(\Dirs,\Ords)$-monotone simultaneous embedding of $\Graphs$ can be constructed in $O(kn^2)$ time.  
\end{theorem}

\begin{proof}
	For the ``only-if'' part, let $\Gamma$ be a $(\Dirs,\Ords)$-monotone simultaneous embedding of $\Graphs$.  Recall that the vertices of $G_i$ ($i=1,\dots,k$), when projected on $\dir_i$, appear in topological ordering $\rho_i$, and therefore, by the definition of $P_i$, in the order they appear in $P_i$. Thus, if we use the position of the vertices of $G_i$ given by $\Gamma$,  and draw the edges of $P_i$ straight-line, we obtain a $\dir_i$-monotone drawing of $P_i$. 
Observe that $\Gamma$ implies a $(\dir_i,\rho_i)$-monotone drawing of $G_i$, thus, by Theorem~\ref{theorem:upse}, $G_i$ has an including planar $st$-digraph $G'_i$ preserving $\rho_i$.  

For the ``if'' part, let $\Gamma$ be a $\Dirs$-monotone simultaneous embedding of $\Paths$ and let $G_i'$ be an including planar $st$-digraph of $G_i$, preserving $\rho_i$, $i=1,\dots,k$. The vertices of $P_i$ (therefore of $G_i'$) appear in the orthogonal projection on $\dir_i$ in the same order as they appear in $P_i$, and therefore in the same order as in~$\rho_i$.  
With this placement of vertices, we can apply Theorem~\ref{theorem:upse} $k$ times, to obtain a $(\dir_i,\rho_i)$-monotone drawing $\Gamma_i$ of $G_i$, $i=1,\dots,k$.   For each vertex $v_j$, all corresponding points $s_{i,j}$ of the drawings $\Gamma_i$, $i=1,\dots,k$, are identical. 
Therefore, the drawings $\Gamma_1,\dots, \Gamma_k$, comprise a $(\Dirs,\Ords)$-monotone simultaneous embedding of $\Graphs$. 
Finally, the claimed running time for the construction follows directly from  Theorem~\ref{theorem:upse}. 
\end{proof}

In the following, using Theorem~\ref{theorem:equiv}, we derive several implications of Theorem~\ref{thm:poly} and Corollary~\ref{corol:threepaths} for upward planar digraphs.
\begin{corollary}
\label{corol:poly}
Given a set $\Dirs=\{\dir_1, \ldots, \dir_k\}$ of vectors and a set $\Graphs=\{G_1,\dots, G_k\}$ of upward planar digraphs together with topological orderings $\Ords=\{\rho_1,\dots,\rho_k\}$,
it can be decided in polynomial time (w.r.t.\ the input size) whether a $(\Dirs,\Ords)$-monotone simultaneous embedding of $\Graphs$ exists. In case of existence, such an embedding can be constructed in polynomial time.
\end{corollary}
\begin{proof}  Let $\Paths=\{P_1,\dots,P_k\}$, where $P_i$ is the path implied by $\rho_i$, $i=1,\dots,k$. 
By Theorem~\ref{theorem:equiv}, $\Graphs$ admits a $(\Dirs,\Ords)$-monotone simultaneous embedding if and only if (1) the set $\Paths$ of paths admits a $\Dirs$-monotone simultaneous embedding and (2) each $G_i$, $i=1,\dots,k$, has an including planar $st$-digraph $G'_i$ preserving $\rho_i$.
By Theorem~\ref{thm:poly}, Condition~(1) can be checked in polynomial time, and the same is true for Condition~(2) by Lemma~\ref{lemma:includingst}.
For the construction, the claimed running time follows from Theorem~\ref{theorem:equiv}. 
\end{proof}

\begin{corollary}
\label{corol:threegraphs}
Given a set  $\{G_1,G_2,G_3\}$ of three upward planar digraphs and their topological orderings $\Ords=\{\rho_1,\rho_2,\rho_3\}$, it can be decided in polynomial time (w.r.t.\ the input size) whether there exists an $\Ords$-constrained simultaneous embedding of $\{G_1,G_2,G_3\}$. If such an embedding exists, it can be constructed in polynomial time as well. 
\end{corollary}
\begin{proof}
	Let $P_i$ be the path implied by $\rho_i$, $i=1,2,3$. By definition,  $\{G_1,G_2,G_3\}$ admits a $\{\rho_1,\rho_2,\rho_3\}$-constrained simultaneous embedding, if and only if there exist vectors  $\dir_1,\dir_2,\dir_3$ such that there exists  a $(\{\dir_1,\dir_2,\dir_3\},\{\rho_1,\rho_2,\rho_3\})$-monotone simultaneous embedding of $\{G_1,G_2,G_3\}$.  By Theorem~\ref{theorem:equiv}, this is equivalent to the facts that: (1) there exist vectors  $\dir_1,\dir_2,\dir_3$ such that the set $\{P_1,P_2,P_3\}$ of paths  admits a $\{\dir_1,\dir_2,\dir_3\}$-monotone simultaneous embedding, and (2) each $G_i$, $i=1,\dots,k$, has an including planar $st$-digraph $G'_i$ preserving $\rho_i$. Conditions~(1) and~(2) can be checked in polynomial time by Corollary~\ref{corol:threepaths} and Lemma~\ref{lemma:includingst}, respectively. By Theorem~\ref{thm_3_always}, vectors $\{\dir_1,\dir_2,\dir_3\}$ can be chosen arbitrarily. Thus, the time complexity of the construction follows directly from Theorem~\ref{theorem:equiv}. 
\end{proof}

Observe that if an upward planar digraph is Hamiltonian\footnote{We say that an upward planar digraph is \emph{Hamiltonian} if it has a directed Hamiltonian path.}, then it has a unique topological ordering. This topological ordering can be found in $O(n)$ time.  Thus, 
from Corollary~\ref{corol:poly} and Corollary~\ref{corol:threegraphs}, we derive the following.

\begin{corollary}
\label{corol:poly_hamiltonian}
Given a set $\Dirs=\{\dir_1, \ldots, \dir_k\}$ of vectors and a set $\Graphs=\{G_1,\dots, G_k\}$ of Hamiltonian upward planar digraphs, 
it can be decided in polynomial time (w.r.t.\ the input size) whether a $\Dirs$-monotone simultaneous embedding of $\Graphs$ exists. If such an embedding exists, it can be constructed in polynomial time.
\end{corollary}

\begin{corollary}
\label{corol:threegraphs_hamiltonian}
Given a set $\Graphs=\{G_1,G_2,G_3\}$ of three Hamiltonian upward planar digraphs, it can be decided in polynomial time (w.r.t.\ the input size) whether $\Graphs$ admits a monotone simultaneous embedding. In case of existence, such an embedding can be constructed in polynomial time. 
\end{corollary}

\section{Complexity issues}\label{sec_hard} 
\subsection{Monotone Simultaneous Embeddings of Paths}\label{sec_hard_1} 

In this section we discuss the relation of monotone simultaneous embeddings of paths to stretchability of pseudolines and realizability of circular sequences. 

Goodman, Pollack, and Sturmfels~\cite{exponential} showed that for each $n$ there exist order types with $n$ elements such that any realization has coordinates that are doubly exponential in the number of points.
Suppose we are given such a realization.
Then, we can add a vertical line between every two consecutive crossings as well as before and after all crossings, and derive a vector and a path for each vertical line.
Given these vectors and paths as an input, our linear program will produce a solution whose binary representation is exponential in the number of lines.
This is, however, no contradiction to the fact that we have a polynomial-time algorithm when we are given the relative distance of the vertical lines: these distances will have a binary representation that is exponential in the number of lines as well, and hence, the algorithm is still polynomial in the input size.
On the other hand, if we are only given the sequence of paths (i.e., the circular sequence of the set), the size of a solution will be exponential in the input size, and hence, there cannot be a polynomial-time algorithm for giving a set of points (assuming a sufficiently strict model of computation).
This fact, however, does not imply intractability of \emph{deciding} the simultaneous embeddability of a sequence of paths, even though the possibly large representation of a realization suggests $\exists \mathbb{R}$-hardness (see~\cite{schaefer}, where these complexity topics are discussed, and, e.g.,~\cite{kratochvil}, where similar issues arise).

With respect to deciding simultaneous embeddability, it is interesting to observe the relation between the problem of stretchability of arrangements of pseudo-lines and our setting.
Obviously, an algorithm deciding stretchability of pseudo-line arrangements also decides whether a simultaneous embedding of a sequence of paths exists.
On the one hand, deciding stretchability is known to be NP-hard~\cite{shor}, and actually equivalent to the existential theory of the reals~\cite{mnev}.
On the other hand, NP-hardness of the stretchability problem does not directly imply NP-hardness of the problem at hand.
However, the problem of deciding whether there exists a point set for a given allowable sequence can easily be reduced to our problem:
add a path for every index sequence of the circular sequence (this corresponds to placing a vertical line directly to the right of each crossing in the corresponding pseudo-line arrangement in the Euclidean plane).
If there exists a simultaneous embedding of this sequence of paths, then the given allowable sequence can be realized.

Note that there is a significant difference between stretchability of a pseudo-line arrangement and realizability of an allowable sequence.
Goodman and Pollack~\cite{non_circular} give an allowable sequence on five elements that is not the circular sequence of any point set, while the smallest non-stretchable pseudo-line arrangement in the projective plane has nine pseudo-lines~\cite{realizable8} (i.e., Ringel's construction using Pappus' Theorem).
We are not aware of any work showing hardness of deciding realizability of allowable sequences (it is not obvious to us that, e.g., Shor's construction~\cite{shor} can also be used in the more constrained setting using allowable sequences, in particular after the transformation to an arrangement with no three pseudo-lines sharing a point).

\subsection{A Closely Related Problem}

The hardness of deciding stretchability of pseudo-lines in the projective plane can be reduced in a straight-forward manner to the following problem closely related to monotone simultaneous embeddings.
Given a sequence $\seq{P_1,\dots,P_k}$ of paths, each containing a subset of arbitrary size of a vertex set $V$, is there a monotone simultaneous embedding such that the directions of monotonicity appear radially around the origin? We call this problem \textsc{General Monotone Simultaneous Embedding} (\emph{GMSE}, for short).
GMSE is more general than our original problem in the sense that  the paths might not contain all the vertices.
Since GMSE requires that the directions of monotonicity are radially ordered, it might seem that GMSE is more restricted than our original problem. However, this is not the case, because for the original monotone simultaneous embedding of a set of paths, only one ordering of the paths has to be considered, as was proven in Proposition~\ref{alowable_set}. However, GMSE is more restricted than the \textsc{Strictly Monotone Trajectory Drawing} problem, which was proven to be NP-hard~\cite{pampel}, 
since in the latter it is not required that the directions of monotonicity are radially ordered.  

 
\begin{proposition}
	\textsc{General Monotone Simultaneous Embedding} (GMSE) is NP-hard.
\end{proposition} 
\begin{proof}
Take any $x$-monotone drawing of the pseudo-line arrangement in the Euclidean plane such that every crossing is at its rightmost possible position.
Let $P_1$ be a path that contains all points in the order given by the line $\dir_1 : x = -\infty$.
Sweep the arrangement in $x$-direction.
At a crossing between the pseudo-lines $s_i$ and $s_j$ (suppose w.l.o.g.\ $I(i,P_1) < I(j,P_1)$), we add the paths $\seq{s_j, s_i, s_l}$ or $\seq{s_l, s_j, s_i}$ for all other elements $s_l$, depending on whether the pseudo-line $s_l$ is above or below the crossing $s_i s_j$.
This clearly encodes the orientation of all triples, and therefore the pseudo-line arrangement in the projective plane is unambiguously fixed.
Also, since the crossings are drawn at their rightmost possible positions, no further restrictions are imposed on the circular sequence of the resulting set (if such a set exists).
If two crossings are independent, we do not care about their relative order, and this is also not captured by the relative order of two paths in which two independent pairs are required to be swapped;
the relative order of the crossings is partial, and moving all crossings to the right as far as possible corresponds to constraining the relative order for each crossing in the highest level of the partial order.
Hence, if and only if there exists a monotone simultaneous embedding of the paths, the pseudo-line arrangement is stretchable.
\end{proof}

\section{Conclusion}\label{sec_conc} 
In this paper we considered monotone simultaneous embeddings of sets of upward planar digraphs, with both predefined and arbitrary directions of monotonicity, where we first concentrated on the special case of directed spanning paths with the same vertex set.

We proved that if a monotone simultaneous embedding of three directed paths exists, then it also exists for an arbitrary choice of directions with the same circular order. 
We also presented a polynomial-time decision- and construction algorithm. Further, we showed that the existence question for an arbitrary number of paths, but with predefined directions, can be solved in polynomial time as well. 

On the other hand, we showed that even if a monotone simultaneous embedding of three given paths exists, it might require an exponential (in the number of vertices) ratio of the smallest and largest distance between points of the embedding. 
Further, we showed that starting from $k=4$, not only the relative circular order of the directions but also the actual choice of the slopes influences monotone simultaneous embeddability.

We also considered the complexity of the problem for $k>3$ paths and arbitrary directions. In contrast to Theorem~\ref{thm:poly}, the construction problem becomes intractable for arbitrary directions, since the constructed embedding might require a representation using coordinates of exponential size. However, showing hardness of the decision question remains an open problem for $k>3$.

We further showed  several implications of the simultaneous embedding of directed paths to upward planar digraphs. However, we had to restrict considerations to the setting where an upward planar digraph is provided together with a topological ordering. The question which remains open is whether our results can be extended to the case where the topological ordering is not provided as a part of the input. 

As mentioned in the introduction, there are three types of simultaneous embeddings, distinguished by the way the edges are represented: geometric simultaneous embeddings, simultaneous embeddings with ``fixed edges'', and simultaneous embeddings without restrictions on the edge representation. We observe that for monotone simultaneous embeddings of paths, it holds that, if the edges can be drawn monotone, then they can also be drawn straight-line. Thus, our monotone simultaneous embeddings of sets of paths are geometric simultaneous embeddings. However, in the case of upward planar digraphs, our algorithm produces a drawing with polygonal path edges for every upward planar digraph. Thus, the same edge might be drawn as different polygonal paths for different upward planar digraphs. In this sense, our monotone simultaneous embeddings of upward planar digraphs are just simultaneous (i.e., without special representation properties of the edges). 
It might be interesting to study what happens if the same edge has to be realized the same way in all of the drawings, i.e., monotone simultaneous embeddings with fixed edges. It seems clear that in this case, if two graphs share an edge, they should have related directions of monotonicity.        

Furthermore, it might also be interesting to consider a problem that is ``between'' the setting with predefined directions and the one with arbitrary directions.
More specifically, let ${\cal A}=\{\alpha_1,\dots,\alpha_k\}$ be a set of wedges centered at the origin. 
If a set $\Paths=\{P_1,\dots,P_k\}$ of paths  admits a $\{\dir_1,\dots,\dir_k\}$-monotone simultaneous embedding, such that $\dir_i \in \alpha_i$, we say that $\Paths$ admits an ${\cal A}$-monotone simultaneous embedding. As a generalization of monotone simultaneous embeddings with fixed directions, it would be interesting to study the computational complexity of deciding whether a set $\Paths$ of paths admits an ${\cal A}$-monotone simultaneous embedding.

\paragraph{Acknowledgements.} 
Research on this topic was 
initiated during a research visit of Sarah Lutteropp and Tamara Mchedlidze in February 2013 in Graz, Austria.
We thank anonymous referees for helpful comments and for making us aware of the related work~\cite{Asinowski2008}.

\bibliographystyle{abbrv}
\bibliography{bibliography}

\end{document}